\documentclass[11pt]{article}
\usepackage{amsmath}
\usepackage{amssymb}
\usepackage{amsbsy}
\usepackage{ifthen}
\usepackage{color}
\usepackage{graphicx}
\usepackage{float}
\usepackage{enumitem}
\usepackage{natbib}
\usepackage{setspace}
\usepackage{amsthm}
\usepackage{fullpage}
\usepackage{multirow}

\newcommand{\ba}{\begin{eqnarray}}
\newcommand{\ea}{\end{eqnarray}}
\newcommand{\nn}{\nonumber}

\newtheorem{Theorem}{Theorem}[section]

\newtheorem{Proposition}[Theorem]{Proposition}

\title{Conditionally Gaussian Random Sequences for an Integrated Variance Estimator with Correlation between Noise and Returns}
\author{Stefano Peluso \footnote{Corresponding author. Universit\`{a} Cattolica del Sacro Cuore, Department of Statistical Sciences and Universit\`{a} della Svizzera italiana, Data Science Lab, ICS. Largo Gemelli 1 20123 Milan. E-mail: stefano.peluso@unicatt.it} 
\and Antonietta Mira \footnote{Universit\`{a} della Svizzera italiana, Data Science Lab, ICS and Università dell'Insubria, 
E-mail: antonietta.mira@usi.ch}
\and Pietro Muliere \footnote{Bocconi University of Milan, E-mail: pietro.muliere@unibocconi.it}}

\begin{document}
\maketitle

\begin{abstract}
Correlation between microstructure noise and latent financial logarithmic returns is an empirically relevant phenomenon with sound theoretical justification. With few notable exceptions, all integrated variance estimators proposed in the financial literature are not designed to explicitly handle such a dependence, or handle it only in special settings. We provide an integrated variance estimator that is robust to correlated noise and returns. For this purpose, a generalization of the Forward Filtering Backward Sampling algorithm is proposed, to provide a sampling technique for a latent conditionally Gaussian random sequence. We apply our methodology to intra-day Microsoft prices, and compare it in a simulation study with established alternatives, showing an advantage in terms of root mean square error and dispersion.

\end{abstract}
\section{Introduction}
Many statistical problems can be formulated as State Space models, where a latent stochastic process $\{\theta_t\}$ evolves in time with dynamics given by a transition equation $\theta_{t+1} = a_1(t)\theta_t + b_1(t)\epsilon_1(t+1)$, and with the process observed noisily through $\{\xi_t\}$, which evolves following the measurement equation $\xi_{t+1} = \tilde A_1(t)\theta_{t+1} + \tilde B_2(t)\epsilon_2(t+1)$, where $\epsilon_1(t)$ and $\epsilon_2(t)$ are Gaussian random variables and $a_1(t)$, $b_1(t)$, $\tilde A_1(t)$ and $\tilde B_2(t)$ are time-varying parameters. \cite{K60} proposed the celebrated Kalman filtering algorithm as optimal solution, in mean square sense, to the \textit{filtering problem}, that is the problem of estimating the unobservable $\theta_t$ by means of observations $\xi^t = \{\xi_1,\dots,\xi_t\}$. 
The Kalman filter is the starting point in \citet{FS94} and \citet{CK94} for an iterative procedure, today commonly known as \textit{Forward Filtering Backward Sampling} (FFBS), for obtaining posterior samples of $\{\theta_t\}$.

\citet{LS72,LS01a,LS01b} introduce the so-called \textit{conditionally Gaussian random sequences}, whose main two features are: (a) dependence of model parameters from past observations or from other random quantities, with the remaining randomness expressed in terms of Gaussian random variables, (b) correlation between $\xi_t$ and $\theta_t$, introduced through the presence, in both the transition and measurement equations, of common Brownian motions.

The Mixture Kalman filter of \citet{CL00} and the Gibbs samplers of \citet{S94a} and \citet{CK96} are some relevant examples that, among other things, generalize Kalman filtering and posterior sampling of the latent stochastic process along direction (a) above. Other simulation techniques such as the Extended Kalman filter of \citet{G72}, the Monte Carlo filter of \citet{K87} and the Particle filter of \citet{GSS93} do not require conditional Gaussianity, but are based on some form of approximation. \citet{G72} provides a suboptimal solution to the filtering problem, by linearizing the transition equation. \citet{K87} and \citet{GSS93} approximate the posterior distribution of the latent stochastic process through a weighted set of particles. A Kalman filter robust to the presence of outliers is proposed in \citet{RSP14}. \citet{DS95} and \citet{DB02} also developed FFBS algorithms. In particular, the methodology of \citet{DS95} defines the conditionally linear Gaussian state space model in terms of a single source of error, and it is able to reproduce correlated shocks between the measurement and the transition equations. Relying on the method of \citet{DS95}, \citet{CS08} develop a new simulation smoother for binomial longitudinal data.

\citet{HS96} and \citet{SK98} point out the empirical relevance of direction (b) for modeling the asymmetric behavior often found in stock prices, and \citet{HW87} emphasize the role of correlation between observed prices and latent stochastic volatility, suggesting that, neglecting it, it can cause significant biases in financial option pricing. Among others, \citet{BK04} and \citet{JPR04} further study this phenomenon in the financial economic literature. Another concrete situation where neglecting correlation in the two equations can be misleading is our motivating example, that is integrated volatility estimation in presence of dependence between microstructure noise and latent financial logarithmic returns, a phenomenon empirically found in \citet{HL06} and theoretically justified in \citet{DS13}. Many integrated variance estimators proposed in the literature (\citealt{ABDL03,AMZ05,ZMA05,Z06,AFX10,BNHLS11,CPA12,PCM12}) are not designed 
to handle such a dependence, except for \citet{BNHLS08}, but only in the special setting of a linear model of endogeneity. To our knowledge, the only papers that deal with this endogenous noise are \citet{KL10}, which propose a robust version of \citet{ZMA05}, and the pre-averaging estimator of \citet{Jacod}. \citet{BR09}, 
despite assuming exogenous noise, still provide a good benchmark method, since it is empirically found to perform well even if the underlying assumptions are violated.

It is now well recognized that the proper use of intra-day financial price observations leads to precise and accurate measurement and forecasting of unobservable measures, through the so-called realized estimators. See, for instance, the beta estimator proposed by \citet{BNHLS11} and \citet{G16}, the realized multivariate covariance of \citet{BNHLS11,PCM12,CPA12} and \citet{ShephardXiu}, the  correlation studied in \citet{BNShephard,AC10} and \citet{BKS13}. In the present paper we propose a realized variance estimator of the daily integrated volatility that is robust to the presence of correlation between microstructure noise and latent returns, generalizing the setting of \citet{AFX10}. For this purpose, we extend the FFBS algorithm from standard State Space models to the more general context of \citet{LS72} in an exact form (with no approximations involved). Therefore our contribution is twofold: (i) the generalization of the FFBS algorithm from State Space models to conditionally Gaussian random sequences, an extension of interest in itself, since it solves the filtering and smoothing problem in a more general context; (ii) the inclusion of the new FFBS algorithm into a MCMC scheme that provides a Bayesian integrated variance estimator robust to correlation between microstructure noise and return, to our knowledge the first Bayesian estimator with these properties. The main advantages of a Bayesian estimator of the integrated variance relying on a system of observational and transition equations are that (i) the latent stochastic price process can be obtained as a byproduct, (ii) from the MCMC iterations any function of the integrated variance or of the latent price process (for instance, integrated quarticity) can be derived, (iii) not only a point estimate, but a whole posterior distribution of the quantity to estimate can be obtained, therefore providing uncertainty quantification of the integrated variance estimate.

The algorithm is presented in Section \ref{Gffbs_sec}, after an introduction to conditionally Gaussian random sequences in Section \ref{Liptser}. The motivating financial problem with related simulated studies and a real application to Microsoft data is detailed in Section \ref{SectionVariance}, and finally the conclusions are drawn in Section \ref{Conclusions}. Matlab codes for the proposed algorithm and the data supporting the findings in this study are available on request from the corresponding author. The data for the empirical application are not publicly available due to privacy restrictions.
\section{Conditionally Gaussian Random Sequences}\label{Liptser}
In this section we introduce the theoretical framework developed in \citet{LS72} (see also \citealt{LS01a} and \citealt{LS01b}), with focus on the recursive equations of conditionally Gaussian random sequences for the solution of the filtering problem.

On a probability space $(\Omega,\mathcal F, P)$, the random sequence $\{\theta_t,\xi_t\}_t$, $t=1,2,\dots$, with $\theta_t=(\theta_1(t),\dots,\theta_k(t))$ and $\xi_t=(\xi_1(t),\dots,\xi_l(t))$, defines the system of recursive equations
\ba
\theta_{t+1} &=& a_0(t,\omega) + a_1(t,\omega)\theta_t + b_1(t,\omega)\epsilon_1(t+1) + b_2(t,\omega)\epsilon_2(t+1) \label{system1a}\\
\xi_{t+1}    &=& A_0(t,\omega) + A_1(t,\omega)\theta_t + B_1(t,\omega)\epsilon_1(t+1) + B_2(t,\omega)\epsilon_2(t+1), \label{system1b} 
\ea
\noindent where $\epsilon_1(t)=(\epsilon_{1,1}(t),\dots,\epsilon_{1,k}(t))$ and $\epsilon_2(t)=(\epsilon_{2,1}(t),\dots,\epsilon_{2,l}(t))$ are independent Gaussian random variables with expected value $\mathbb E(\epsilon_{i,j}(t))=0$ and $\mathbb E(\epsilon_{i_1,j_1}(t)\epsilon_{i_2,j_2}(s)) = \delta(i_1,i_2)\delta(j_1,j_2)\delta(t,s)$, for all $i$ and $j$, where
$$\delta(x,y) = \left\{ \begin{array}{ll} 1, & x = y \\ 0, & x \neq y \end{array} \right.$$
In the sequel, $\theta_t$ and $\xi_t$ are, respectively, unobservable and observed vectors, with $\theta_0|\xi_0 \sim \Phi(m,\gamma)$, that is Gaussian with mean $m$ and variance $\gamma$. $a_0(t,\omega)$ and $A_0(t,\omega)$ are vector functions, and $a_1(t,\omega)$, $A_1(t,\omega)$, $b_1(t,\omega)$, $b_2(t,\omega)$, $B_1(t,\omega)$ and $B_2(t,\omega)$ are matrix functions, square integrable and measurable at time $t$. All the vector and matrix functions at time $t$ are collected in $D(t,\omega)$. In \citet{LS72}, $D(t,\omega)$ is assumed to be $\mathcal F_t^{\xi}$-measurable, where $\mathcal F_t^{\xi} = \sigma\{\omega:\ \xi_0,\dots,\xi_t\}$ is the $\sigma$-algebra generated by the random variables $\xi_0,\dots,\xi_t$. This assumption will be relaxed in Section \ref{Gffbs_sec}, where measurability with respect to $\sigma$-algebras generated by other random variables will be considered. Denote by $b \circ b = b_1b_1^*+b_2b_2^*$, $b \circ B = b_1B_1^*+b_2B_2^*$ and $B \circ B = B_1B_1^*+B_2B_2^*$ where $X^*$ is the transposed matrix of $X$ and $X^+ = Y^*(YY^*)^{-2}Y$ is the pseudo-inverse matrix of $X$, with $Y$ such that $Y^*Y=X$. For ease of notation we suppress the dependence on $\omega$.
\begin{Theorem}\label{Theorem1}
Suppose that $\mathbb E(||\theta_0||^2 + ||\xi_0||^2) < \infty$, $|(a_1(t))_{ij}|< L$ and $|(A_1(t))_{ij}|<L$, where $L$ is a positive constant. Then, $\theta_t|\xi_0,\dots,\xi_t \sim \Phi(m(t),\gamma(t))$, where $m(t)$ and $\gamma(t)$ are determined from the recursive equations
\ba
m(t+1)      &=& [a_0(t) + a_1(t)m(t)] + [b \circ B(t) + a_1(t)\gamma(t)A_1^*(t)] \nn\\
             && \cdot [B \circ B(t) + A_1(t)\gamma(t)A_1^*(t)]^+ \cdot [\xi_{t+1} - A_0(t) - A_1(t)m(t)] \label{filtering1a}\\
\gamma(t+1) &=& [a_1(t)\gamma(t)a_1^*(t) + b \circ b(t)] - [b \circ B(t) + a_1(t)\gamma(t)A_1^*(t)] \nn\\
             && \cdot [B \circ B(t) + A_1(t)\gamma(t)A_1^*(t)]^+ \cdot [b \circ B(t) + a_1(t)\gamma(t)A_1^*(t)]^* \label{filtering1b}         
\ea
\noindent with the initial conditions $m(0) = m$ and $\gamma(0) = \gamma$. 
\end{Theorem}
\begin{proof}
See \citet{LS72}, Theorem 3.2.
\end{proof}
An important special case is when $D(t)$ is not a random, but a deterministic function of time $t$. In this case, if the vector $(\theta_0,\xi_0)$ is Gaussian, the process $(\theta_t,\xi_t)$ will also be Gaussian, with known covariance $\gamma(t)$. 
In this setting it is possible to reformulate the system of recursive equations \eqref{system1a} and \eqref{system1b} so that the dependence between $\xi_t$ and $\theta_t$ is explicit, and to recover the Kalman filter as special case.

The random sequence $\{\theta_t,\xi_t\}_t$ is known as conditionally Gaussian since it follows a Gaussian distribution at any specific time $t$, conditionally on the knowledge of $D(t)$. Note that this is not a restrictive assumption, since unconditionally the dependence in time and space is not necessarily linear (for instance when the distribution of $a_1(t)$ depends on $\theta_t$), and the disturbances are location-scale mixture of Gaussian random variables. A wide class of continuous distributions may be constructed as location-scale mixture of Normal distributions, such as contaminated Normals, Student's t, Logistic, Laplace and Stable distributions. As specified in \citet{MW92}, one way of seeing that the class of Normal mixture densities is very broad results by recalling that any density, even strongly multi-modal and asymmetric, can be approximated arbitrarily well by a Normal mixture. This is a setting of interest in finance, where we often observe skewed distributions of returns (see, among others, \citealt{B97} and \citealt{AC03}). Furthermore, distributions of returns can be contaminated by outliers that are not easy to detect and correct for, and that can severely distort a non robust estimation methodology, causing for instance relevant consequences on asset allocation studies (\citealt{BG92}). Finally, as pointed out in \citet{ES99}, multi-modal distributions can model situations of regime switches, known to have a relevance in option pricing (see for instance \citealt{BE02}) and mean-variance portfolio selection (\citealt{ZY03}, among others). 
\section{Sampling Algorithm of the Latent Process} \label{Gffbs_sec}
System \eqref{system1a} and \eqref{system1b} can be reformulated to highlight the relation between $\xi_t$ and $\theta_t$, so that the sequence of the observations can be interpreted as a realization of a stochastic Markovian latent process with measurement noise: 
\ba
\theta_{t+1} &=& a_0(t) + a_1(t)\theta_t + b_1(t)\epsilon_1(t+1) + b_2(t)\epsilon_2(t+1) \label{system2a}\\
\xi_{t+1}    &=& \tilde A_0(t) + \tilde A_1(t)\theta_{t+1}+ \tilde B_1(t)\epsilon_1(t+1) + \tilde B_2(t)\epsilon_2(t+1), \label{system2b} 
\ea
\noindent where $a_0(t),a_1(t),b_1(t),b_2(t),\tilde A_0(t),\tilde A_1(t),\tilde B_1(t),\tilde B_2(t)$ are stored in $\tilde D(t)$. This alternative representation is more common in the econometrics, financial and engineering literature, and it can be derived from the system \eqref{system1a} and \eqref{system1b} since, substituting \eqref{system2a} in \eqref{system2b}, $\xi_{t+1}$ can be written as
\ba
\xi_{t+1}    &=& \tilde A_0(t) + \tilde A_1(t)[a_0(t) + a_1(t)\theta_t + b_1(t)\epsilon_1(t+1) \nn\\
&& + b_2(t)\epsilon_2(t+1)] + \tilde B_1(t)\epsilon_1(t+1) + \tilde B_2(t)\epsilon_2(t+1) \nn,
\ea
\noindent clarifying that the relation between $D(t)$ and $\tilde D(t)$ is given by
\ba
\left\{\begin{array}{l} A_0(t) = \tilde A_0(t) + \tilde A_1(t)a_0(t)\\
                     A_1(t) = \tilde A_1(t)a_1(t)\\
                     B_1(t) = \tilde A_1(t)b_1(t) + \tilde B_1(t)\\
                     B_2(t) = \tilde A_1(t)b_2(t) + \tilde B_2(t).\\ \end{array} \right. \label{relation1}
\ea
Given the system \eqref{system2a}-\eqref{system2b}, from Theorem \ref{Theorem1} it follows that $\theta_t|\xi_1,\dots,\xi_t \sim \Phi(m(t),\gamma(t))$, where $m(t)$ and $\gamma(t)$ are obtained by the recursive equations \eqref{filtering1a} and \eqref{filtering1b}, but with $A_0(t)$, $A_1(t)$, $B_1(t)$ and $B_2(t)$ replaced by the respective right hand sides in \eqref{relation1}. When $b_2(t)=0$ and $\tilde B_1(t)=0$ (or, equivalently, when $b_1(t)=0$ and $\tilde B_2(t)=0$) for all $t$, system \eqref{system2a}-\eqref{system2b} simplifies to 
\ba
\theta_{t+1} &=& a_0(t) + a_1(t)\theta_t + b_1(t)\epsilon_1(t+1) \label{system3a}\\
\xi_{t+1}    &=& \tilde A_0(t) + \tilde A_1(t)\theta_{t+1} + \tilde B_2(t)\epsilon_2(t+1), \label{system3b} 
\ea
for which the filtering problem can be solved through the Kalman filtering iterations:
\ba
m(t+1)      &=& [a_0(t) + a_1(t)m(t)] + [a_1(t)\gamma(t)A_1^*(t)] \nn\\
             && \cdot [B_2^2(t) + A_1(t)\gamma(t)A_1^*(t)]^+ \cdot [\xi_{t+1} - A_0(t) - A_1(t)m(t)] \nn\\
\gamma(t+1) &=& [a_1(t)\gamma(t)a_1^*(t) + b_1^2(t)] - [a_1(t)\gamma(t)A_1^*(t)] \nn\\
             && \cdot [B_2^2(t) + A_1(t)\gamma(t)A_1^*(t)]^+ \cdot [a_1(t)\gamma(t)A_1^*(t)]^*. \nn         
\ea
In the simplified setting of model \eqref{system3a}-\eqref{system3b}, \citet{FS94} and \citet{CK94} introduce the Forward Filtering and Backward Sampling (FFBS) algorithm, to sample $\theta^T$ a posteriori from 
$$ p(\theta^T|\xi^T,\tilde D(1),\dots,\tilde D(T)) \propto \prod_{t=1}^T \phi(V_t W_t, V_t),$$
\noindent where 
\ba
V_t^{-1} &=& A_1^*(t) (B_2^2(t))^+ A_1(t) + a_1^*(t) (b_1^2(t))^+ a_1(t) + \gamma^+(t) \nn\\
W_t &=& A_1(t) (B_2^2(t))^+ (\xi_{t+1}-A_0(t)) + a_1^*(t)(b_1^2(t))^+(\theta_{t+1}-a_0(t)) + \gamma^+(t)m(t). \nn
\ea

Exploiting an extended factorization of the posterior density of $\theta$, induced by the shared Brownian motions, we derive a generalized version of the FFBS algorithm, to jointly sample $$\theta_1,\dots,\theta_T|\xi_1,\dots,\xi_T,\tilde D(1),\dots,\tilde D(T)$$ from the system \eqref{system2a}-\eqref{system2b} (an equivalent algorithm for the system \eqref{system1a}-\eqref{system1b} can also be formulated). For easier reference in the sequel, we refer to this algorithm as G-FFBS.
\begin{Proposition}\label{Prop1}
Given $\xi^T$ generated from model \eqref{system2a}-\eqref{system2b}, then 
$$ p(\theta^T|\xi^T) \propto \prod_{t=1}^T \phi(V_t W_t, V_t),$$
\noindent where 
\ba
V_t^{-1} &=& (A_1(t)-B\circ b(t) (b\circ b)^+(t) a_1(t))^* \Sigma_t^+ (A_1(t)-B\circ b(t) (b\circ b)^+(t) a_1(t)) \nn\\
&& + a_1^*(t) (b\circ b)^+(t) a_1(t) + \gamma^+(t) \nn\\
W_t &=& (A_1(t)-B\circ b(t) (b\circ b)^+(t) a_1(t))^* \Sigma_t^+ (\xi_{t+1}-A_0(t)\nn\\
&& - B\circ b(t)(b\circ b)^+(t)(\theta_{t+1}-a_0(t))) + a_1^*(t)(b\circ b)^+(t)(\theta_{t+1}-a_0(t)) + \gamma^+(t)m(t) \nn\\
\Sigma_t &=& B \circ B(t) - B \circ b(t) (b \circ b)^+(t) B \circ b^*(t)\nn.
\ea
\end{Proposition}
\begin{proof}
	See Appendix A.
\end{proof}
The proposed generalization over the traditional FFBS finds relevant empirical justification in the motivating example that will be discussed in Section \ref{SectionVariance}. The algorithm requires a forward step in which the quantities of interest $m(t)$ and $\gamma(t)$ are computed following Theorem \ref{Theorem1}, and a backward step where the latent process is sampled according to the factorization in \eqref{factorization}. In the traditional FFBS algorithm, the factor at time $t$ in \eqref{factorization} is proportional to $p(\theta_{t+1}|\theta_t) p(\theta_t|\xi^t)$, whilst in the proposed G-FFBS algorithm, there is an additional term $p(\xi_{t+1}|\theta_{t+1},\theta_t)$, since the correlation between measurement and transition errors generates a conditional dependence between $\xi_{t+1}|\theta_{t+1}$ and $\theta_t$. When $\tilde B_1(t) = 0$ and $b_2(t)=0$ for all $t$ or when $\tilde B_2(t) = 0$ and $b_1(t)=0$, there is no correlation between the two errors, the conditional independence of the observations is restored, and G-FFBS reduces to FFBS. 

For posterior inference on any function of the latent stochastic process $g(\theta^T)$, three cases can be distinguished: (i) $\tilde D(t)$ is measurable at time $t$, (ii) $\tilde D(t)$ is unknown at time $t$ but with known dynamics, (iii) $\tilde D(t)$ is unknown at time $t$ and with unknown dynamics. In case (i), $\tilde D(t)$ is measurable at time $t$ with respect to the $\sigma$-algebra generated by $\xi^T$ or by some other observables, and all samples from $\theta^T|\xi^T$ can be obtained through the G-FFBS. In case (ii) a simple procedure for posterior inference requires to recursively estimate $\tilde D(t)$ by $\hat D(t)$, which is estimated by the known dynamics, and then use $\hat D(t)$ instead of $\tilde D(t)$ in the G-FFBS (see \citealt{SW83} and \citealt{CMP01} for, respectively, a biometric and a financial application). When in (iii), $\tilde D(t)$ is unknown and cannot be parametrically forecasted: a complete Bayesian model has to be specified, with prior $\pi(\tilde D(1),\dots,\tilde D(T))$, and MCMC procedures are used to sample from the joint posterior $\mathbb P(\theta^T,\tilde D(1),\dots,\tilde D(T)|\xi^T)$, by repeatedly sampling at each iteration 
\begin{itemize}
\item $\mathbb P(\theta^T|\tilde D(1),\dots,\tilde D(T),\xi^T)$,
\item $\mathbb P(\tilde D(1),\dots,\tilde D(T)|\theta^T,\xi^T) \propto \mathbb P(\theta^T,\xi^T|\tilde D(1),\dots,\tilde D(T)) \, \pi(\tilde D(1,\omega),\dots,\tilde D(T,\omega)).$
\end{itemize}
The first step is executed through G-FFBS, and the whole algorithm is a Gibbs sampler (\citealt{GG84,GS90}) or a Metropolis-Hastings sampler (\citealt{MRRTT53,H70}), depending on wheather  $\pi(\tilde D(1),\dots,\tilde D(T))$ is a conjugate prior or not.

We conclude this section with a note on model parameters identifiability. If proper priors are adopted, in a Bayesian setting different values of parameters corresponding to the same likelihood value do not arise identifiability issues, with the exception of degenerate cases when the prior and the posterior distribution concide. To better understand this point let us collect in $\{\tilde{D}(t)\}$ all parameters $\tilde D(1),\dots,\tilde D(T))$. If, for different values of $ \{\tilde{D}(t)\}$, say $ \{\tilde{D}_1(t)\}$ and $ \{\tilde{D}_2(t)\}$, $\mathbb P(\theta^T,\xi^T|\{\tilde{D}_1(t)\})$ and $\mathbb P(\theta^T,\xi^T|\{\tilde{D}_2(t)\})$ are the same, there are no identifiability problems as long as $\mathbb P(\{\tilde{D}(t)\}|\xi^T)$ differs from $\mathbb P(\{\tilde{D}(t)\})$ for at least one value of $\{\tilde{D}(t)\}$. If $\mathbb P(\theta^T,\xi^T|\{\tilde{D}_1(t)\})=\mathbb P(\theta^T,\xi^T|\{\tilde{D}_2(t)\})$ and also $\mathbb P(\{\tilde{D}_1(t)\})=\mathbb P(\{\tilde{D}_2(t)\})$, we can only conclude that $\{\tilde{D}(t)\}$ has the same posterior probability  in correspondence of $\{\tilde{D}_1(t)\}$ and $\{\tilde{D}_2(t)\}$, but still $\{\tilde{D}(t)\}$ has a proper posterior distribution. The case of $\mathbb P(\{\tilde{D}(t)\}|\xi^T)=\mathbb P(\{\tilde{D}(t)\})$ occurs when the data does not provide any information on $\{\tilde{D}(t)\}$, a degenerate case verified only when $\mathbb P(\xi^T|\{\tilde{D}(t)\})$ is constant for all values of $\{\tilde{D}(t)\}$. 
\section{Robust Integrated Variance Estimation}\label{SectionVariance}
\subsection{Problem context}
In this section the developed sampling algorithm is applied to our motivating problem. Suppose that the logarithmic price of a given financial asset follows, within the trading day, the diffusion process
\ba
d\theta_t = c(t)dZ_t \nn
\ea
\noindent where $c(t)$ is the instantaneous volatility and $\{Z_t\}_t$ is the standard Brownian motion. $IV = \int c^2(t)dt$ is known as \textit{integrated variance} and is of interest as a measure of the true daily volatility. For estimation we use the discrete approximation of the continuous-time process above: $\theta_{(t+1)/T} = \theta_{t/T} + c_{t/T}\sqrt{1/T}Z_t$, where we have restated the time subscripts of the trading day in the interval $[0,1]$, $T^{-1}$ is the discrete time interval between adjacent observations, $\theta_{t/T}-\theta_{(t-1)/T}=O_p(T^{-1/2})$ and $Z_t$ is a standard Gaussian. $IV$ is a latent quantity, usually estimated with the so-called \textit{realized variance} $RV=\sum_{t=1}^T (\theta_{t/T}-\theta_{(t-1)/T})^2$, the sum of all intra-day high frequency observed logarithmic returns. $RV$ is a consistent and efficient estimator of $IV$ (\citealt{ABDL03}) when there is no microstructure noise, that is when $\theta_{t/T}$ for $t=1,\dots,T$ is directly observed. When microstructure noise is introduced, we observe $\xi_{t/T}$ instead of $\theta_{t/T}$, and the computable realized variance becomes $\tilde {RV}=\sum_{t=1}^T (\xi_{t/T}-\xi_{(t-1)/T})^2$. Note that we do not specify the continuous-time version of the measurement equation: the observed price relates to the latent price only through the microstructure noise, consequence of trades occurring at discrete times. Unfortunately, $\tilde {RV}$ loses the good properties of $RV$, since it is biased and inconsistent for the true integrated variance. As this problem arises mainly when the frequency of observations approaches infinity (that is when the maximum distance between adjacent measurement times approaches zero), it can be attenuated by sparse sampling, but this involves a loss of information because of the discarded data. Recently, some authors have followed the approach suggested by \citet{AMZ05} of sampling as often as possible and modeling the noise. In particular, a first consistent estimator of $IV$ for financial data contaminated by microstructure noise has been proposed in \citet{ZMA05} (whose order of convergence is improved in \citealt{Z06}), later followed by \citet{BNHLS08}, that propose a kernel-based estimator. There have been numerous extensions of the framework with noisy observations that account for additional empirically observed data irregularities, as asyncronicity of multivariate log prices, serially dependent microstructure noise, positivity of the estimator, skewness and kurtosis, presence of outliers, lead-lag effects  (see, for instance, \citealt{GT91,AFX10,BNHLS11,CPA12,PCM12,HRV14,buccheri2018high}). Less attention has been posed on the dependence between microstructure noise and latent financial logarithmic returns, empirically found in \citet{HL06}. Also, common microstructure theories from financial economics literature justify a correlation between latent returns and microstructure noise (\citealt{DS13}) by the presence of uninformed trades, risk aversion and market makers learning speed. All the estimators mentioned above are not designed for such a dependence, except for \citet{BNHLS08}, but only for a linear model of endogeneity. \citet{KL10} robustifies the estimator of \citet{ZMA05} to the presence of endogenous noise, and \citet{Jacod} propose a generalized pre-averaging estimator of the integrated variance accounting for various noise structures. The kernel estimator of \citet{BR09} also shows robustness properties that justify its adoption in a setting with correlation between microstructure noise and latent returns. 

\subsection{The proposed estimator}
The framework of conditionally Gaussian sequences, with the sampling algorithm introduced above, can be used to propose a new estimator of the integrated variance that is robust to the presence of correlation between microstructure noise and latent returns. Consider the bivariate system
\ba
\xi_{(t+1)/T} &=& \theta_{(t+1)/T} +\tilde B_1(t)\epsilon_1(t+1)+\tilde B_2(t)\epsilon_2(t+1) \label{system3aa}\\
\theta_{(t+1)/T} &=& \theta_{t/T} + b_1(t)\epsilon_1(t+1), \label{system3bb}
\ea
\noindent in which $a_0(t)=\tilde A_0(t)=b_2(t)=0$ and $a_1(t)=\tilde A_1(t)=1$ for all $t$. Model \eqref{system3aa}-\eqref{system3bb} is completed by characterizing the prior distributions: $\tilde B_1(t) \sim \phi(\mu_{B,t},\sigma^2_{B,t})$, $b_1(t) \sim \phi(\mu_{b,t},\sigma^2_{b,t})$ and finally $\tilde B_2(t) \sim IG(\alpha_{B,t},\beta_{B,t})$. The correlation between microstructure noise and true returns is introduced through the random variable $\epsilon_1$, appearing in both the equations. Note that \citet{HL06} found microstructure noise and latent returns negatively correlated: with a Gaussian prior on $B_1(t)$ it is possible to center, a priori, this correlation on a negative value. Furthermore, \citet{DS13} point out that a negative correlation appears more realistic, and that markets with no evidence of significant negative correlation are likely subject to an
extraordinary microstructure effect such as high risk aversion.

The full conditional distribution of $\theta^T$ is sampled with the G-FFBS. The forward step of the G-FFBS algorithm is performed through the following filtering iterations:
\ba
m(t+1) &=& m(t)+\frac{b_1(t)B_1(t)+\gamma(t)}{B_1^2(t)+B_2^2(t)+\gamma(t)} (\xi_{(t+1)/T}-m(t)) \nn\\
&=& m(t)+\frac{b_1(t)(b_1(t)+\tilde B_1(t)) +\gamma(t)}{(b_1(t)+\tilde B_1(t))^2 +\tilde B_2^2(t)+\gamma(t)} (\xi_{(t+1)/T}-m(t)) \label{filtering2a}\\
\gamma(t+1) &=& (\gamma(t)+b_1^2(t)) - \frac{(b_1(t)B_1(t)+\gamma(t))^2}{B_1^2(t)+B_2^2(t)+\gamma(t)} \nn\\
&=& (\gamma(t)+b_1^2(t)) - \frac{(b_1(t)(b_1(t)+\tilde B_1(t))+\gamma(t))^2}{(b_1(t)+\tilde B_1(t))^2+\tilde B_2^2(t)+\gamma(t)}. \label{filtering2b}
\ea
Note that if $\tilde B_1(t)=0\ \forall t$, the filtering iterations \eqref{filtering2a} and \eqref{filtering2b} simplify to the Kalman filter iterations (\citealt{K60}): 
\ba
m(t+1) &=& m(t)+\frac{b_1^2(t) +\gamma(t)}{b_1^2(t) +\tilde B_2^2(t)+\gamma(t)} (\xi_{(t+1)/T}-m(t)) \nn\\
\gamma(t+1) &=& (\gamma(t)+b_1^2(t)) - \frac{(b_1^2(t)+\gamma(t))^2}{b_1^2(t)+\tilde B_2^2(t)+\gamma(t)}.\nn \label{filtering2bb}
\ea

For the backward sampling step, $\theta^T|\xi^T$ are sampled from \eqref{factorization}, where 
\ba
p(\theta_{t/T}|\theta_{(t+1)/T},\dots,\theta_{T},\xi^T) &\propto& p(\xi_{(t+1)/T}|\theta_{(t+1)/T},\theta_{t/T}) p(\theta_{(t+1)/T}|\theta_{t/T}) p(\theta_{t/T}|\xi^t)\nn\\
&=& \phi\left(\xi_{(t+1)/T};\theta_{t/T}+\frac{B_1(t)}{b_1(t)}(\theta_{(t+1)/T}-\theta_{t/T}), B_2^2(t))\right) \nn\\
&& \phi\left(\theta_{(t+1)/T};\theta_{t/T},b_1^2(t)\right) \phi\left(\theta_{t/T};m(t),\gamma(t)\right) \nn\\
&\propto& \phi(V_tW_t,V_t), \label{eqApp}
\ea
\noindent with $V_t$ and $W_t$ defined in Appendix B.

Note that $B_1(t)=\tilde B_1(t)+b_1(t)$ and $B_2(t)=\tilde B_2(t)$. The correlation between transition and measurement error can be removed by fixing $\tilde B(t)=0$. In this case, $B_1(t)=b_1(t)$ and, as expected, $V_t=\left(1-\frac{\gamma(t)}{b_1^2(t)+\gamma(t)}\right)\gamma(t)$ and $W_tV_t=\left(1-\frac{\gamma(t)}{b_1^2(t)+\gamma(t)}\right)m(t)+\frac{\gamma(t)}{b_1^2(t)+\gamma(t)}\theta_{(t+1)/T}$, as in the usual FFBS.

\subsection{Some properties of the estimator}
The difference between FFBS and G-FFBS can be crucial for the estimation of the latent stochastic process. We highlight that the result in \eqref{eqApp} serves the purpose of sampling the latent stochastic process, and therefore the implied realized variance, from its correct posterior distribution under the general setting of conditionally Gaussian random sequences. Therefore, under our modeling assumptions, the consistency to the correct values is guaranteed by the MCMC properties. Unbiasedness in finite sample is not assured, unless one implements appropriately built unbiased MCMC schemes \citep{jacob2017unbiased}, which is beyond the scope of our paper. In finite samples we can say that the estimate of the integrated variance is optimal in the mean square sense, that is no other estimator can have a lower mean square error under our modeling assumptions, since the posterior mean is also the solution to the smoothing problem of conditionally Gaussian random sequences, solution known to be optimal in the mean square sense \citep{LS01b}. 

To study the asymptotic FFBS bias in a simplified setting, in this section we assume that in the model for observations and latent process expressed in Equations \eqref{system3aa} and \eqref{system3bb} the parameter values are constant in $t$ or they eventually stabilize to some steady state, starting from some value of $t$. Then for all $t=1,\dots,T$, $\tilde{B}_1(t)=\tilde{B}_1$, $\tilde{B}_2(t)=\tilde{B}_2$ and $b_1(t)=b_1$, with $\gamma$ converging to
\ba
\gamma^* := \frac12 \left( \sqrt{(2\tilde{B}_1+b_1)^2+4\tilde{B}_2^2} - (2\tilde{B}_1+b_1) \right), 
\ea 
which reduces to $\gamma_0^*:=\frac12 \left( \sqrt{b_1^2+4\tilde{B}_2^2} - b_1 \right)$ when correlation is neglected. We can assume the existence and uniqueness of such a limit since the conditions for asymptotic properties of the optimal linear filtering are satisfied (Theorem 14.3 of \citealt{LS01b}). Ignoring correlation results in a negative asymptotic bias if $V_t$, computed for the model with no correlation, is lower than the corresponding quantity in the full model. Equivalently, looking at the functional form of $V_t$ in Appendix B, the asymptotic negative bias resulting from neglecting the correlation occurs when 
$$
\frac{\left(1-\frac{\tilde B_1+b_1}{b_1}\right)^2}{\tilde B_2^2} + \gamma^{*-1} < \gamma_0^{*-1},
$$
which, after some algebra, can be written as
\ba
\frac1{b_1^2}\tilde{B}_1^4 + \frac2{b_1}\tilde{B}_1^3 > \frac{\sqrt{b_1^2+4\tilde{B}_2^2}}{b_1}\tilde{B}_1^2 + \left(b_1+\sqrt{b_1^2+4\tilde{B}_2^2}\right)\tilde{B}_1. \label{eq:ineq}
\ea

For specific annualized values of $b_1$ and $\tilde{B}_2$, the difference between $V_t$ computed with and without correlation is shown in Figure \ref{Bias}. Omitting the correlation implies a negative bias in correspondence of $\tilde B_1$ values at which the black solid line $\tilde{B}_1^4/b_1^2 + 2\tilde{B}_1^3/b_1$ is above the red dashed line $(\sqrt{b_1^2+4\tilde{B}_2^2}/b_1)\tilde{B}_1^2 + (b_1+\sqrt{b_1^2+4\tilde{B}_2^2})\tilde{B}_1$ , and a positive bias vice-versa. Therefore the direction of the asymptotic bias tends to follow the sign of $\tilde{B}_1$, with the exception of more extreme negative or positive $\tilde{B}_1$, for which the bias direction is reversed. Also note that the distortion is not symmetric for negative and positive $\tilde{B}_1$.
\begin{figure}[ht]
	\begin{center}
		\includegraphics[width=0.5\columnwidth]{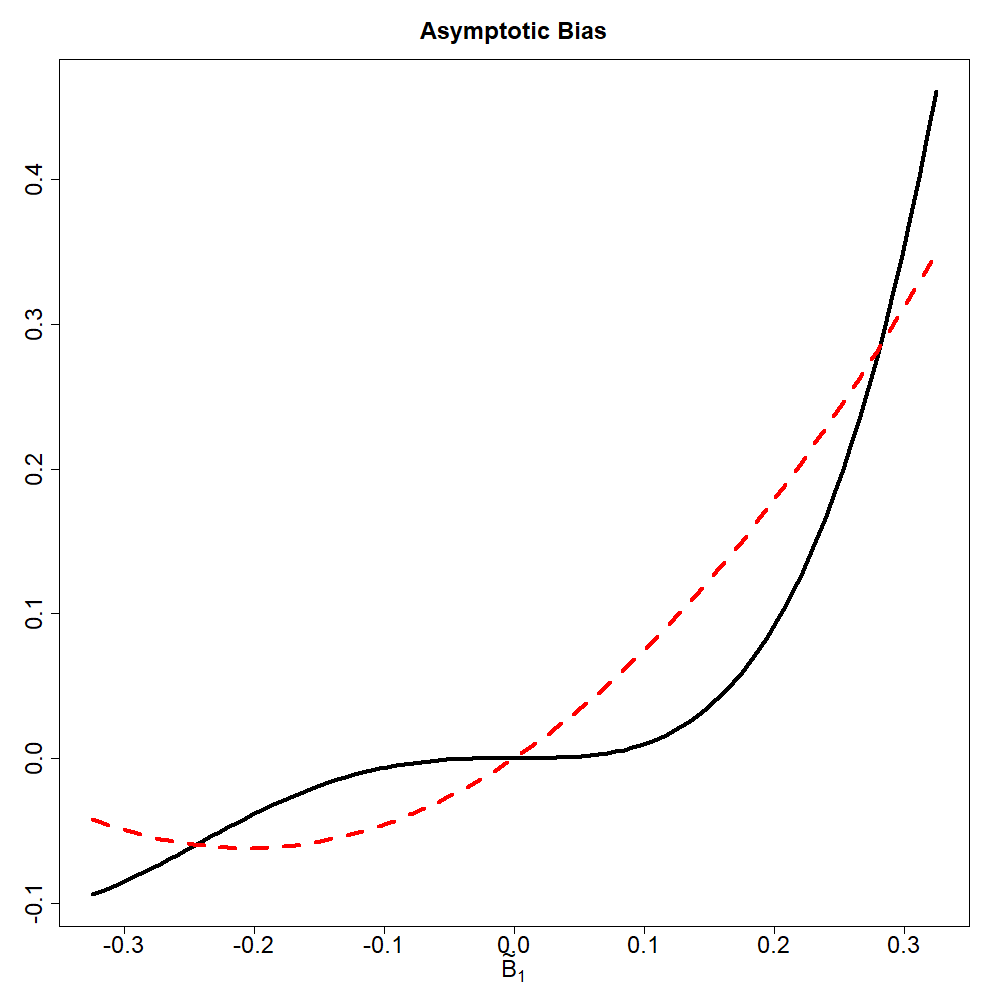}
	\end{center}
	\caption{Graphical representation of the asymptotic bias resulting when the model ignores the correlation between microstructure noise and latent returns. The black solid line is $\tilde{B}_1^4/b_1^2 + 2\tilde{B}_1^3/b_1$, whilst the red dashed line is $(\sqrt{b_1^2+4\tilde{B}_2^2}/b_1)\tilde{B}_1^2 + (b_1+\sqrt{b_1^2+4\tilde{B}_2^2})\tilde{B}_1$. Regions of black solid line above the red dashed line indicate negative bias, otherwise positive bias.}
	\label{Bias}
\end{figure}

For instance, for a correlation $\rho$ between microstructure noise and financial latent return taking values in the set $\pm\{0.15,0.30,0.75,0.90\}$, a noise to signal ratio ($NTS$) of 1.5 and an annualized transition error variance of 0.06, we simulate, for each value of $\rho$, 500 trading days, with $T=23400$ seconds per business day. To fix the correlation to the desired level, we generate the data imposing $\tilde B_2=\sqrt{(1-\rho^2)b_1^2\cdot NTS}$ and $\tilde B_1=\sqrt{\rho^2 b_1^2\cdot NTS}$ (scenarios with positive correlation) or $\tilde B_1=-\sqrt{\rho^2 b_1^2\cdot NTS}$ (scenarios with negative correlation). In this way, $\rho = sgn(b_1)\tilde{B}_1/(\sqrt{\tilde{B}_1^2+\tilde{B}_2^2})$.
For each day we compute the estimated quadratic variation for FFBS and G-FFBS, that is the sum of the squared first differences in $\theta_{1/T},\theta_{2/T},\dots,\theta_1$ sampled from distribution in \eqref{eqApp} (G-FFBS) and from \eqref{eqApp} with $\tilde{B}_1=0$ (FFBS), and we compare them in Figure \ref{QV}. It is clear that neglecting $\rho$ has an impact on the inference of the latent process. As expected, the distance between the two methodologies widens in the magnitude of the correlation: see in the left figure how FFBS worsens with higher and higher negative correlations introduced in the system, against a G-FFBS algorithm that remains unbiased. But, as expected from \eqref{eq:ineq} and its graphical representation in Figure \ref{Bias}, the FFBS bias direction does not necessarily follow the sign of the correlation: negative correlation is imposed through a negative $\tilde{B}_1$, but in the case of $\rho=-0.90$, the annualized $\tilde B_1=-0.27$ is outside the region $(-0.245,0)\cup(0.281,\infty)$ for which the bias would be negative. The results are similar in the right panel, when positive correlations of 0.15, 0.75 and 0.9 are hypothesized: more and more correlation worsens the quadratic variation estimated by FFBS, but, as expected, asymmetrically relative to the scenarios with negative correlation: the impact of a higher correlation seems worse, and in the most extreme scenario with $\rho=0.9$, the bias does not become negative since $\tilde{B}_1=0.27$, inside the region $(-\infty,-0.245)\cup(0,0.281)$ of positive FFBS bias.
\begin{figure}[ht]
\begin{center}
\includegraphics[width=\columnwidth]{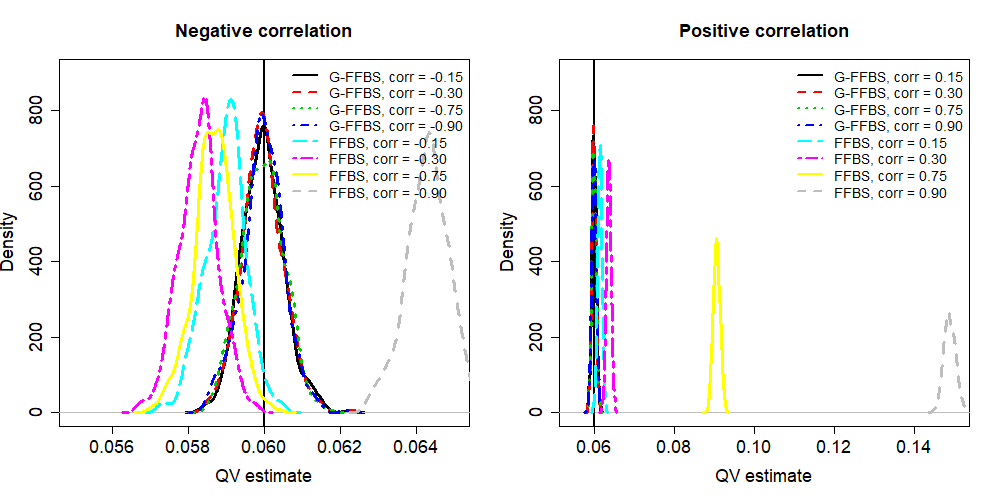}
\end{center}
\caption{Estimated Quadratic Variation simulated with FFBS and G-FFBS and the true latent value (vertical bar), over 500 trading days, when the correlation between microstructure noise and financial latent return is -0.15, -0.30, -0.75 and -0.90 (left plot) or 0.15, 0.30, 0.75 and 0.90 (right plot).}
\label{QV}
\end{figure}

\subsection{Other MCMC steps}
To sample from the remaining full conditional distributions, note that 
\ba
\begin{pmatrix}\xi_{(t+1)/T} \\ \theta_{(t+1)/T}\end{pmatrix} | \theta_{t},b_1(t),B_1(t),B_2(t)\sim \Phi\left\{\begin{pmatrix}\theta_{t} \\ \theta_{t}\end{pmatrix},\begin{pmatrix}B_1^2(t)+B_2^2(t) && b_1(t)B_1(t) \\ b_1(t)B_1(t) && b_1^2(t)\end{pmatrix}\right\}, \nn 
\ea
\noindent and 
\ba
\xi_{(t+1)/T}| \theta_{(t+1)/T},\theta_{t},b_1(t),B_1(t),B_2(t) &\sim& \phi\left(\theta_{t/T}+\frac{B_1(t)}{b_1(t)}(\theta_{(t+1)/T}-\theta_{t/T}),B_2^2(t)\right). \nn
\ea
The full conditionals of $\tilde B_1(t)$ and $\tilde B_2^2(t)$ are in standard form and provided in Appendix B. On the other hand, we sample $b_1(t)$ with a \textit{Hamiltonian} step (see Chapter 5 in \citealt{BGJM11} 
for an introduction to the algorithm). The motivation for using this step is its ability to exploit the information in the full conditional gradient of $b_1(t)$, for a faster exploration of the parameter space, thus overcoming the random walk behavior of the Metropolis-Hastings step in a highly dimensional space. We refer the Reader to Appendix C for the details on the Hamiltonian step.
Note that, when there is no correlation (that is when $\tilde B_1(t)=0$), the sampler can be reduced to the Gibbs algorithm in \citet{PCM12}.

The output of the whole algorithm is a collection of samples $$\{\theta^T_{(i)},\tilde B^T_{1(i)},\tilde B^T_{2(i)},b^T_{1(i)}\}_{i=1}^M,$$ \noindent where $M$ is the number of iterations of the MCMC scheme. Then the proposed estimator of the integrated variance is 
\ba
\frac{1}{M-M_0} \sum_{i=M_0+1}^M \sum_{t=1}^T (\theta_{(t+1)/T,(i)}-\theta_{t,(i)})^2\label{IVestimator}
\ea 
\noindent where $M_0<M$ is the burn-in, that is the number of samples discarded at the beginning of the MCMC chain to allow the simulation process to reach its stationary regime. To summarize, the procedure for obtaining the IV estimator is the following:
\textit{
\begin{enumerate}
	\item For iterations $i=1\dots,M$
	\begin{enumerate}
	\item Sample $\theta^T_{(i)}$ from the G-FFBS algorithm in Proposition \ref{Prop1}, assuming  $a_0(t)=\tilde A_0(t)=b_2(t)=0$ and $a_1(t)=\tilde A_1(t)=1$ for all $t$
	\item Sample  $\tilde B^T_{1(i)}$ from the full conditional \eqref{fullB1} in Appendix B
	\item Sample $\tilde B^T_{2(i)}$ from the full conditional \eqref{fullB2} in Appendix B
	\item Sample $b^T_{1(i)}$ from the Hamiltonian step highlighted in Appendix C
	\end{enumerate}	
\item Compute the estimator given in Equation \eqref{IVestimator}.
\end{enumerate}
}

We simulate 500 trading days, for $M=1000$, $M_0=500$ and correlations $\pm0.10$, starting all the chains from values significantly different from the true ones. The hyper-parameters are $\mu_{B,t}=-1.48\cdot 10^{-5}$, $\sigma^2_{B,t}=1.53\cdot 10^{-10}$, $\mu_{b,t}=1.21\cdot 10^{-4}$, $\sigma^2_{b,t}=1.02\cdot 10^{-8}$, $\alpha_{B,t}=2.1$ and $\beta_{B,t}=1.99\cdot 10^{-8}$ for all $t$, fixed so that they are at least 20\% higher or lower than the true values used to generate the datasets. Our methodology is compared with the estimators of \citet{KL10}, \citet{BR09} and \citet{Jacod} (for \citealt{Jacod}, both the adjusted and unadjusted estimators for small sample sizes are implemented). 
For completeness, we add to the comparison other popular estimators, as the quasi-maximum likelihood estimator of \citet{xiu2010quasi}, the realized kernel of \citet{BNHLS11}, and the two-scale estimator proposed by \cite{ZMA05}. 
The method of \citet{KL10} requires the choice of the tuning parameter $K$: we use $K=T^{2/3}$, since it performs well in the simulations in \citet{KL10} and it is what the authors suggest in their empirical study. Alternative values of $K$ are shown in \citet{KL10} to perform worse and depend on unobservable quantities estimated with a slow-decaying bias. For the estimator proposed in \citet{BR09}, the tuning parameters are chosen according to the rule of thumb proposed in Equation (26) of \citet{BR09}, in simulation computed using the true values and in the application below to Microsoft Corporation, using the corresponding values in Table 1 of \citet{BR06}. Finally, the tuning parameters of \citet{Jacod} are fixed, using their notation, to $k_n=51$, $\theta=k_n/\sqrt{T}$ and $g(x)=x \wedge (1-x)$, as in their simulation studies. The results are reported in Figure \ref{HMC} and in Table \ref{Table1}: there is a clear advantage for our methodology in terms of dispersion and root mean square error (RMSE). The quasi-maximum likelihood estimator performs particularly well in terms of bias, even if it shows some relevant positive dispersion that contributes to increase the RMSE to a level higher than that of the method we propose.
\begin{figure}[ht]
\begin{center}
\includegraphics[scale=0.5]{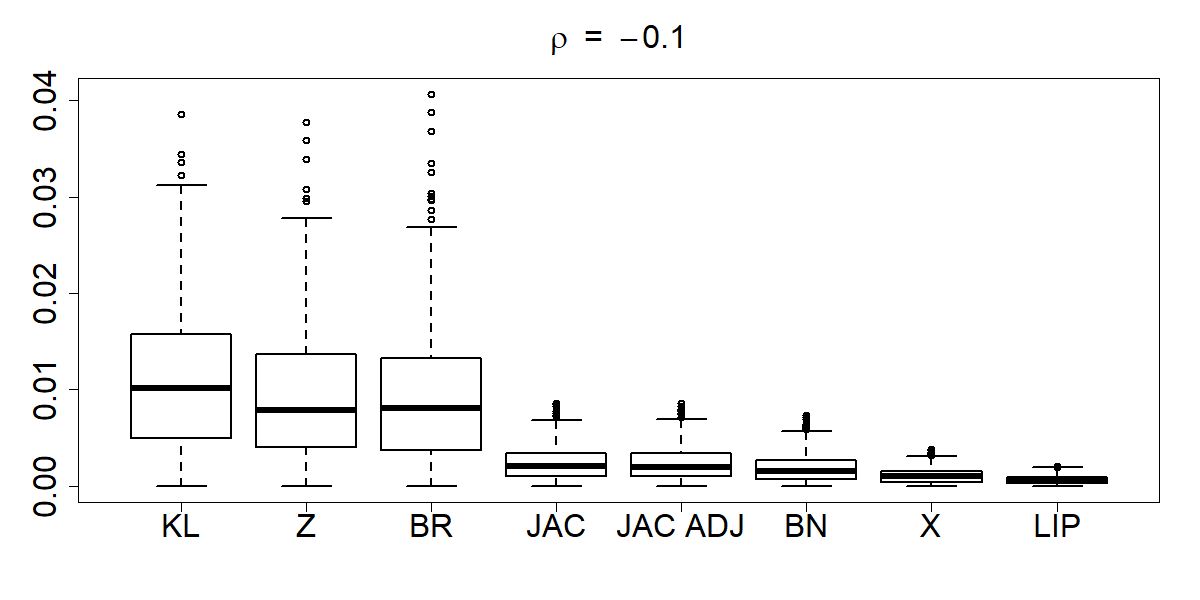}\\
\includegraphics[scale=0.5]{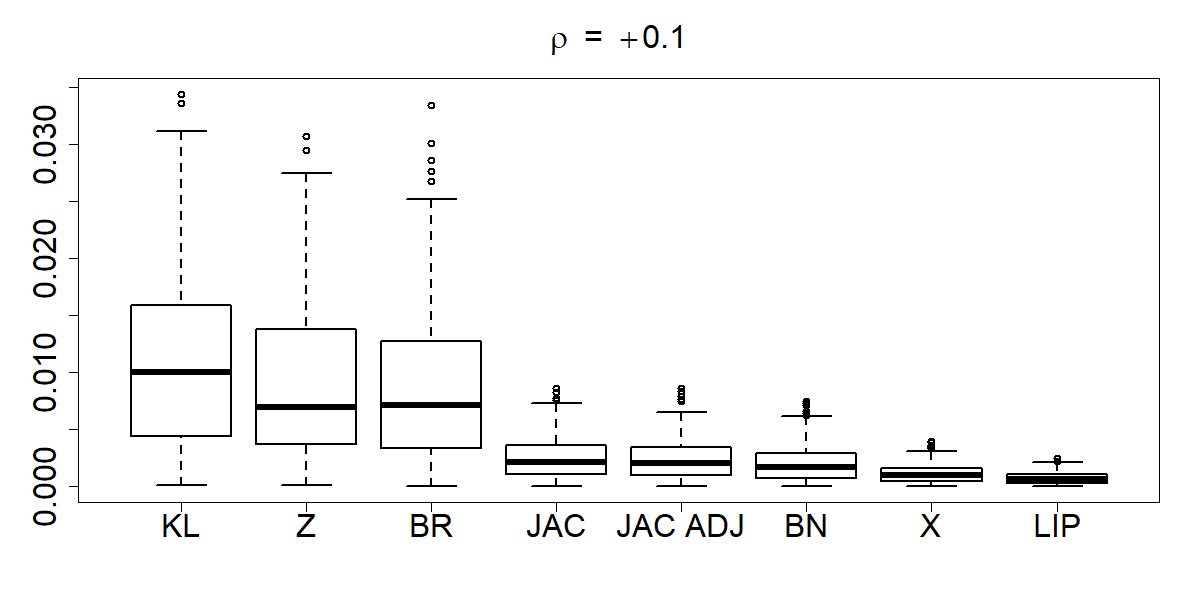}
\end{center}
\caption{Absolute differences between estimated quadratic variations and the true latent value, for the methods in \citet{KL10} (KL), \citet{ZMA05} (Z), \citet{BR09} (BR), \citet{Jacod} (JAC), the small sample adjusted estimator of \citet{Jacod} (JAC ADJ), \citet{BNHLS11} (BN), \citet{xiu2010quasi} (X) and for our methodology (LIP). The correlation between microstructure noise and financial latent return is fixed to $\pm0.10$.}
\label{HMC}
\end{figure}

\begin{table}[ht]
	\centering
	\small
	\begin{tabular}{|l|ccc|ccc|}
		\hline
				& \multicolumn{3}{|c|}{$\rho = -0.10$} & \multicolumn{3}{|c|}{$\rho = +0.10$}\\		
				Method & Bias $\times$ 1000 & Std $\times$ 1000 & RMSE $\times$ 1000 & Bias $\times$ 1000 & Std $\times$ 1000 & RMSE $\times$ 1000\\ 
		\hline
		KL & -4.53 & 12.62 & 13.41 & -5.42 & 11.93 & 13.11 \\ 
		Z & -2.74 & 11.37 & 11.70& -3.24 & 10.72 & 11.19 \\ 
		BR & -1.02 & 11.72 & 11.76& -1.53 & 11.04 & 11.14 \\ 
		JAC & -0.39 & 2.97 & 3.00& -0.26 & 3.01 & 3.03 \\ 
		JAC ADJ & -0.09 & 2.99 & 2.99& 0.04 & 3.03 & 3.03 \\ 
		BN & 0.89 & 2.35 & 2.52& 1.19 & 2.41 & 2.69 \\ 
		X & \textbf{0.03} & 1.35 & 1.35& \textbf{0.10} & 1.40 & 1.40 \\ 
		LIP & 0.56 & \textbf{0.56} & \textbf{0.80}& \textbf{-0.10} & \textbf{0.90} & \textbf{0.90} \\ 
		\hline
	\end{tabular}
\caption{Bias, standard deviation and RMSE for the methods in \citet{KL10} (KL), \citet{ZMA05} (Z), \citet{BR09} (BR), \citet{Jacod} (JAC), the small sample adjusted estimator of \citet{Jacod} (JAC ADJ), \citet{BNHLS11} (BN), \citet{xiu2010quasi} (X) and for our methodology (LIP), over 500 trading days, in the simulation setting with correlation between microstructure noise and financial latent return fixed to $\pm0.10$.}
\label{Table1}
\end{table}

We also run the algorithm on 1-second frequency logarithmic prices of Microsoft Corporation, for the period April 1, 2007 - June 30, 2008, and the estimated annualized quadratic variations are reported in Figure \ref{IV_real}. A practical implication of the differences in the estimation of Microsoft integrated variances is a Gaussian Value At Risk that deviates, on average over the period studied, from 2\% to 6\% of a hypothetical initial investment.

\begin{figure}[ht]
\begin{center}
\includegraphics[scale=0.5,angle=-90]{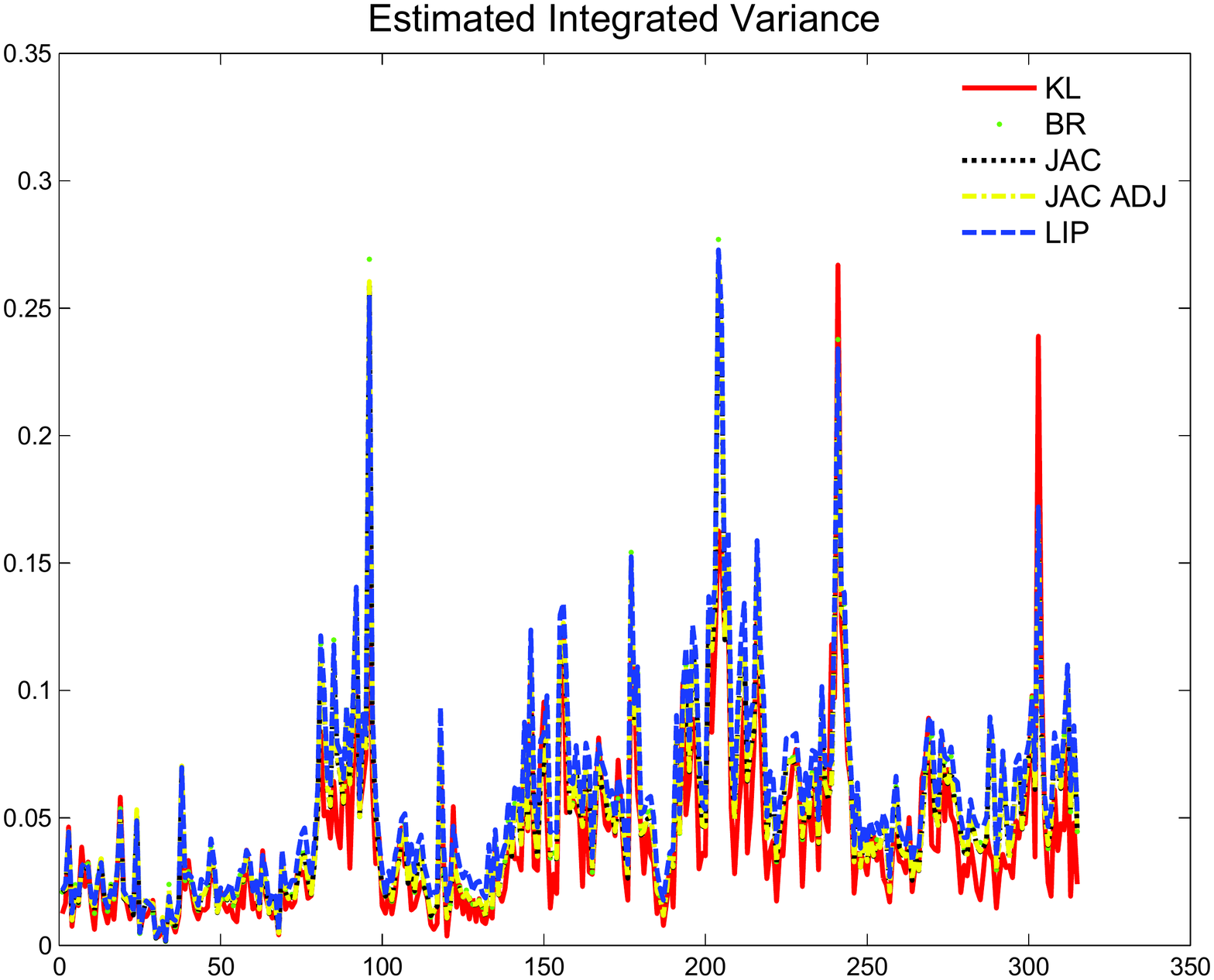}
\end{center}
\caption{Integrated variance for Microsoft Corporation, estimated with the methods in \citet{KL10} (KL), \citet{BR09} (BR), unadjusted (JAC) and adjusted (JAC ADJ) \citet{Jacod}, and with our methodology (LIP), in the period April 1, 2007 - June 30, 2008.}
\label{IV_real}
\end{figure}
\section{Conclusions}\label{Conclusions}
Overwhelming evidence contrasts the independent microstructure noise assumption, in favour of market noise correlated with increments in the efficient price, with important implications for volatility estimation based on high-frequency data (\citealt{HL06}). Furthermore, such a dependence naturally arises in common microstructure models, as discussed in depth in \citet{DS13}. On the other hand, with the notable exceptions of \citet{BNHLS08}, \citet{KL10} and \citet{ZMA05}, several results in the literature analyze high-frequency volatility estimation assuming that the noise process is independent of the efficient price. In the present paper we use the theoretical framework of the conditionally Gaussian random sequences of \citet{LS72,LS01a,LS01b}, 
to propose a new integrated variance estimator that is robust to correlation between microstructure noise and latent returns. To this aim, we adopt a Bayesian perspective and sample a posteriori the latent price process through a generalization of the Forward Filtering Backward Sampling algorithm of \citet{FS94} and \citet{CK94}. An application to Microsoft 1-second logarithmic prices is provided, and a simulation study shows an improved performance of our estimator in terms of RMSE and dispersion, relative to the alternatives in the literature. Our methodology can be implemented in other financial problems, for instance to generalize the framework of \citet{B97} to normal inverse Gaussian financial logarithmic returns with measurement error, or, following the approaches of \citet{HRS92} and \citet{HRS94}, in ARCH and Stochastic Volatility models.

\section*{Acknowledgements}
Stefano Peluso acknowledges support from the Swiss National Science Foundation (SNF), Grant No. P1TIP1\_155031. Antonietta Mira also gratefully acknowledges the financial support by SNF.

\section*{Appendix A: Proof of Proposition \ref{Prop1}}
	Using the notation $x^t = \{x_1,\dots,x_t\}$ and suppressing the dependence on $\tilde D(1),\dots,\tilde D(T)$, G-FFBS exploits the factorization 
	\ba
	p(\theta^T|\xi^T) = \prod_{t=1}^T p(\theta_t|\theta_{t+1},\dots,\theta_T,\xi^T). \label{factorization}
	\ea
	Noting that
	\ba
	&& p(\theta^T,\xi^T) = \prod_{t=1}^T p(\theta_t,\xi_t|\theta_{t-1}) \nn\\
	&& \ \ = \prod_{t=0}^{T-1} \phi\left\{\begin{pmatrix}\xi_{t+1} \\ \theta_{t+1}\end{pmatrix}; \begin{pmatrix} A_0(t) + A_1(t)\theta_{t} \\ a_0(t)+a_1(t)\theta_{t}\end{pmatrix},\begin{pmatrix} B \circ B(t) && b \circ B(t) \\ (b \circ B(t))^* && b \circ b(t) \end{pmatrix}\right\}, \nn
	\ea
	\noindent $\xi_{t+1}|\theta_{t+1},\theta_{t} \sim \phi(\mu_t,\Sigma_t)$, where 
	\ba
	\mu_t &=& A_0(t)+A_1(t)\theta_t + B \circ b(t) (b \circ b)^+(t) (\theta_{t+1}-a_0(t)-a_1(t)\theta_t). \nn
	\ea
	\noindent Thus the factor $p(\theta_t|\theta_{t+1},\dots,\theta_T,\xi^T)$ in \eqref{factorization} can be expressed as
	\ba
	&& p(\theta_t|\theta_{t+1},\dots,\theta_T,\xi^T) \propto p(\theta_t,\dots,\theta_T,\xi_{t+1},\dots,\xi_{T}|\xi^t) \nn\\
	&&\ \ = \prod_{i=t+1}^T p(\xi_i,\theta_i|\theta_{i-1}) \cdot p(\theta_t|\xi^t)\nn\\
	&&\ \ \propto p(\xi_{t+1}|\theta_{t+1},\theta_t) p(\theta_{t+1}|\theta_t) p(\theta_t|\xi^t)\nn\\
	&&\ \ \propto \exp\left\{-\frac12\left((\xi_{t+1}-\mu_t)^*\Sigma^+_t(\xi_{t+1}-\mu_t)\right)\right\}\nn\\
	&&\ \ \ \ \ \cdot \exp\left\{-\frac12\left((\theta_{t+1}-a_0(t)-a_1(t)\theta_t)^*(b\circ b)^+(t)(\theta_{t+1}-a_0(t)-a_1(t)\theta_t) \right)\right\}\nn\\
	&&\ \ \ \ \ \cdot \exp\left\{-\frac12\left((\theta_{t}-m(t))^*\gamma^+(t)(\theta_{t}-m(t)) \right)\right\}\nn\\
	&&\ \ \propto \exp\left\{-\frac12\left[\mu_t^*\Sigma^+_t\mu_t-2\mu_t^*\Sigma^+_t\xi_{t+1} + \theta_{t}^*a_1(t)^*\left(b\circ b\right)^+a_1(t)\theta_{t} +\right.\right. \nn\\
	&&\ \ \ \ \ \left.\left. -2\theta_t^*a_1(t)^*\left(b\circ b\right)^+\left(\theta_{t+1}-a_0(t) + \theta_t^*\gamma^+(t)\theta_t - 2\theta_t^*\gamma^+(t)m(t)\right) \right]\right\}\nn\\
	&&\ \ \propto \exp\left\{-\frac12\left[\theta_t^*\left(A_1(t)-B\circ b(b\circ b)^+a_1(t)\right)^*\Sigma_t^+\left(A_1(t)-B\circ b(b\circ b)^+a_1(t)\right)\theta_t + \right.\right.\nn\\
	&&\ \ \ \ \ \left.\left. +2\theta_t^*\left(A_1(t)-B\circ b(b\circ b)^+a_1(t)\right)^*\Sigma_t^+\left(A_0(t)+B\circ b(b\circ b)^+(\theta_{t+1}-a_0(t))\right)  + \right.\right.\nn\\
	&&\ \ \ \ \ \left.\left. -2\theta_t^*\left(A_1(t)-B\circ b(b\circ b)^+a_1(t)\right)^*\Sigma_t^+\xi_{t+1}  +\theta_{t}^*a_1(t)^*\left(b\circ b\right)^+a_1(t)\theta_{t} +\right.\right. \nn\\
	&&\ \ \ \ \ \left.\left. -2\theta_t^*a_1(t)^*\left(b\circ b\right)^+\left(\theta_{t+1}-a_0(t) + \theta_t^*\gamma^+(t)\theta_t - 2\theta_t^*\gamma^+(t)m(t)\right) \right]\right\}\nn\\
	&&\ \ = \exp\left\{-\frac12\left( \theta_t^*V_t^{-1}\theta_t - 2\theta_t^*W_t \right)\right\}\nn\\ 
	&&\ \ \propto \phi\Big(V_t W_t,V_t\Big).\nn
	\ea

\section*{Appendix B: Auxiliary quantities and full conditionals not mentioned in the main text}
Quantities $V_t$ and $V_tW_t$ for Equation \eqref{eqApp}:
\ba
V_t &=& \left(\frac{\left(1-\frac{B_1(t)}{b_1(t)}\right)^2}{B_2^2(t)}+\frac1{b_1^2(t)}+\frac1{\gamma(t)}\right)^{-1} \nn\\
&=& \left( 1-\frac{\left(1-\frac{B_1(t)}{b_1(t)}\right)^2b_1^2(t)\gamma(t) + B_2^2(t)\gamma(t)}{\left(1-\frac{B_1(t)}{b_1(t)}\right)^2b_1^2(t)\gamma(t) + B_2^2(t)\gamma(t)+B_2^2(t)b_1^2(t)}\right) \gamma(t) \nn
\ea
\ba
V_tW_t &=& V_t\left[\frac{1-\frac{B_1(t)}{b_1(t)}}{B_2^2(t)}\left(\xi_{(t+1)/T}-\frac{B_1(t)}{b_1(t)}\theta_{(t+1)/T}\right) + \frac{\theta_{(t+1)/T}}{b_1^2(t)} + \frac{m(t)}{\gamma(t)}\right] \nn\\
&=& \left( 1-\frac{\left(1-\frac{B_1(t)}{b_1(t)}\right)^2b_1^2(t)\gamma(t) + B_2^2(t)\gamma(t)}{\left(1-\frac{B_1(t)}{b_1(t)}\right)^2b_1^2(t)\gamma(t) + B_2^2(t)\gamma(t)+B_2^2(t)b_1^2(t)}\right)m(t) + \nn\\ 
&& \frac{\left(1-\frac{B_1(t)}{b_1(t)}\right)\left(\frac{\xi_{(t+1)/T}}{\theta_{(t+1)/T}}-\frac{B_1(t)}{b_1(t)}\right)b_1^2(t)\gamma(t) + B_2^2(t)\gamma(t)}{\left(1-\frac{B_1(t)}{b_1(t)}\right)^2b_1^2(t)\gamma(t) +B_2^2(t)\gamma(t)+B_2^2(t)b_1^2(t)}\theta_{(t+1)/T}.\nn
\ea

Full conditionals of $\tilde B_1(t)$ and $\tilde B_2^2(t)$ in Section \ref{SectionVariance}:
\ba
&& p\left(\tilde B_1(t)|\theta^T, \xi^T,\tilde B_2^T,b_1^T,\{\tilde B_1(s),\ s \neq t\}\right) \nn\\
&&\ \ \ \propto p(\xi_{(t+1)/T}|\theta_{(t+1)/T},\theta_{t/T},\tilde B_1(t),\tilde B_2(t),b_1(t)) p(\tilde B_1(t)) \nn\\
&&\ \ \ \propto \frac{\phi\left\{\frac{\mu_{B,t}+\frac{\sigma^2_{B,t}}{b_1(t)\tilde B_2^2(t)}(\theta_{(t+1)/T}-\theta_{t/T})(\xi_{(t+1)/T}-\theta_{(t+1)/T})}{\sqrt{1+\frac{\sigma^2_{B,t}}{b_1^2(t)\tilde B_2^2(t)}(\theta_{(t+1)/T}-\theta_{t/T})^2}},\sigma^2_{B,t}\right\}}{\sqrt{1+\frac{\sigma^2_{B,t}}{b_1^2(t)\tilde B_2^2(t)}(\theta_{(t+1)/T}-\theta_{t/T})^2}} \label{fullB1}\\
&& \nn\\
&& p\left(\tilde B_2^2(t)|\theta^T, \xi^T,\tilde B_1^T,b_1^T,\{\tilde B_2(s),\ s \neq t\}\right) \nn\\
&&\ \ \ \ \propto p(\xi_{(t+1)/T}|\theta_{(t+1)/T},\theta_{t/T},\tilde B_1(t),\tilde B_2^2(t),b_1(t)) p(\tilde B_2^2(t)) \nn\\
&&\ \ \ \ \propto IG\left\{\alpha_{B,t}+\frac12,\beta_{B,t}+\frac12\left(\xi_{(t+1)/T}-\theta_{(t+1)/T}-\frac{\tilde B_1(t)}{b_1(t)}(\theta_{(t+1)/T}-\theta_{t/T})\right)^2\right\}\label{fullB2}
\ea

\section*{Appendix C: Hamiltonian step of Section \ref{SectionVariance}}
The Hamiltonian step is performed through the following iterative procedure:
\begin{enumerate}
\item Sample the auxiliary momentum variable $p\{1\}$ from $\Phi(0,1)$,
\item Propose $b_1(t)^*$ from the Leapfrog algorithm. In details, fix $k\{1\}$ to the current value of $b_1(t)$. For step size $\epsilon$ and number of iterations $L$:
\small
\ba
&& p\{1+\epsilon/2\} = p\{1\} - \frac\epsilon2 \left[\frac{k\{1\}-\mu_{b,t}}{\sigma^2_{b,t}} \right.\nn\\
&&\ \ \ \ \ \ \left. + \frac1{\tilde B_2(t)^2} \left(\xi_{(t+1)/T}-\theta_{(t+1)/T} - \frac{\tilde B_1(t)}{k\{1\}} (\theta_{(t+1)/T}-\theta_{t/T}) \right) \frac{\tilde B_1(t)}{k\{1\}^2} (\theta_{(t+1)/T}-\theta_{t/T})\right] \nn
\ea   
\normalsize    
\noindent For $i=1,\dots,L-1$:
\small
\ba
k\{1+i\epsilon\} &=& k\{1+(i-1)\epsilon\} + \epsilon p\{1+(i-1/2)\epsilon\} \nn\\
p\{1+(i+1/2)\epsilon\} &=& p\{1+(i-1/2)\epsilon\} - \epsilon \left[\frac{k\{1+i\epsilon\}- \mu_{b,t}}{\sigma^2_{b,t}} + \right.\nn\\
&& \left.\frac1{\tilde B_2(t)^2} \left(\xi_{(t+1)/T}-\theta_{(t+1)/T} - \frac{\tilde B_1(t)}{k\{1+i\epsilon\}} (\theta_{(t+1)/T}-\theta_{t/T}) \right)\frac{\tilde B_1(t)}{k\{1+i\epsilon\}^2} (\theta_{(t+1)/T}-\theta_{t/T})\right] \nn
\ea
\normalsize
Finally,
\small
\ba
k\{1+L\epsilon\} &=& k\{1+(L-1)\epsilon\} + \epsilon p\{1+(L-1/2)\epsilon\} \nn\\
p\{1+L\epsilon\} &=& p\{1+(L-1/2)\epsilon\} - \frac{\epsilon}2 \left[\frac{k\{1+L\epsilon\}- \mu_{b,t}}{\sigma^2_{b,t}} + \right.\nn\\
&& \left.\frac1{\tilde B_2(t)^2} \left(\xi_{(t+1)/T}-\theta_{(t+1)/T} - \frac{\tilde B_1(t)}{k\{1+L\epsilon\}} (\theta_{(t+1)/T}-\theta_{t/T}) \right)\frac{\tilde B_1(t)}{k\{1+L\epsilon\}^2} (\theta_{(t+1)/T}-\theta_{t/T}) \right] \nn
\ea
\normalsize
\noindent and the proposed value is $b_1(t)^*=k\{1+L\epsilon\}$.
\item Evaluate potential and kinetic energies $U$ and $Z$ at proposed and current values:
\small
\ba
U(t) &\propto& \frac{(b_1(t)-\mu_{b,t})^2}{2\sigma^2_{b,t}} + \frac1{2\tilde B_2(t)^2} \left[\xi_{(t+1)/T}-\theta_{(t+1)/T} - \frac{\tilde B_1(t)}{b_1(t)} (\theta_{(t+1)/T}-\theta_{t/T})\right]^2\nn\\
Z(t) &=& \frac12 p\{1\}^2 \nn\\
U(t)^* &\propto& \frac{(b_1(t)^*-\mu_{b,t})^2}{2\sigma^2_{b,t}} + \frac1{2\tilde B_2(t)^2} \left[\xi_{(t+1)/T}-\theta_{(t+1)/T} - \frac{\tilde B_1(t)}{b_1(t)^*} (\theta_{(t+1)/T}-\theta_{t/T})\right]^2\nn\\
Z(t)^* &=& \frac12 p\{1+L\epsilon\}^2 \nn
\ea
\normalsize
\item Accept $b_1(t)^*$ with probability 
$$\min\left( 1, \exp\{U(t)-U(t)^*+Z(t)-Z(t)^*\} \right).$$
\end{enumerate}

\small
\singlespacing
\bibliographystyle{apalike}
\bibliography{refLiptser}
\normalsize

\end{document}